\documentclass[journal,onecolumn,12pt]{IEEEtran}

\usepackage{graphicx}
\usepackage{amssymb}
\usepackage{epstopdf}
\usepackage[cmex10]{amsmath}
\usepackage{amsthm}
\usepackage{enumerate}
\usepackage{eufrak}
\usepackage{cite}
\usepackage{mathcomp}
\usepackage{supertabular}
\usepackage{longtable}
\usepackage{stmaryrd}
\usepackage{url}
\usepackage{color}
\usepackage{rotating}
\usepackage{float}
\usepackage[boxed]{algorithm2e}
\usepackage{array}
\usepackage{placeins}
\usepackage{multicol}
\usepackage{setspace}
\usepackage{circuitikz}

\newtheorem{proposition}{Proposition}[section]
\newtheorem{theorem}{Theorem}[section]
\newtheorem{corollary}{Corollary}[section]
\newtheorem{lemma}{Lemma}[section]
\newtheorem{definition}{Definition}[section]

\newtheorem{example}{Example}[section]
\newtheorem{construction}{Construction}[section]

\newcommand{\F}{\mathcal {F}}
\newcommand{\A}{\mathcal {A}}

\renewcommand{\H}{\mathcal {H}}
\newcommand{\bz}{{\Bbb Z}}
\newcommand{\B}{{\cal B}}
\newcommand{\C}{{\cal C}}
\newcommand{\D}{{\cal D}}
\newcommand{\E}{{\cal E}}
\newcommand{\CH}{{\cal H}}

\newcommand{\CP}{{\cal P}}
\newcommand{\CT}{{\cal T}}

\newcommand{\su}{{\sf u}}
\newcommand{\sv}{{\sf v}}

\newcommand{\vn}{{\mathbf{n}}}
\newcommand{\vw}{{\mathbf{w}}}
\newcommand{\vX}{{\mathbf{X}}}
\newcommand{\vx}{{\mathbf{x}}}
\newcommand{\vy}{{\mathbf{y}}}
\newcommand{\vk}{{\mathbf{k}}}
\newcommand{\vt}{{\mathbf{t}}}
\newcommand{\vT}{{\mathbf{T}}}
\newcommand{\vA}{{\mathbf{A}}}
\newcommand{\vB}{{\mathbf{B}}}

\ifCLASSINFOpdf
\else
\fi
\begin{document}
%
\title{Constructions of Optimal and Near-Optimal Multiply Constant-Weight Codes}

\author{Yeow~Meng~Chee,~\IEEEmembership{Senior Member, IEEE},
        Han~Mao~Kiah,
        Hui~Zhang,
        and~Xiande~Zhang
\thanks{Y. M. Chee, H. Zhang and X. Zhang
are with the Division~of~Mathematical Sciences,
  School~of~Physical~and~Mathematical~Sciences,
  Nanyang~Technological~University, 21~Nanyang~Link, Singapore~637371,
  Singapore (emails:\{YMChee, huizhang, XiandeZhang\}@ntu.edu.sg).}
\thanks{Han Mao Kiah is with Coordinated Science Lab,
University of Illinois, Urbana-Champaign, 1308 W. Main Street,
Urbana, IL 61801, USA,
    (email:{hmkiah@illinois.edu}).}
}

\maketitle

\begin{abstract}
Multiply constant-weight codes (MCWCs) have been recently studied to
improve the reliability of certain physically unclonable function
response. In this paper, we give combinatorial constructions for
MCWCs which yield several new infinite families of optimal MCWCs.
Furthermore, we demonstrate that the Johnson type upper bounds of MCWCs are
asymptotically tight for fixed weights and distances.
Finally, we provide bounds and constructions of two dimensional MCWCs.
\end{abstract}

\begin{IEEEkeywords}
multiply constant-weight codes, physically unclonable functions,
generalized packing designs, two dimensional multiply
constant-weight codes.
\end{IEEEkeywords}

%
\IEEEpeerreviewmaketitle

\section{Introduction}

A multiply constant-weight code (MCWC) is a binary code of length
$mn$ where each codeword is partitioned into $m$ equal parts and has
weight exactly $w$ in each part
\cite{Cheeetal:MCWC,CDGKS}.
This definition
therefore generalizes the class of {\em constant-weight codes}
(where $m=1$) and a subclass of {\em doubly constant-weight codes}
(where $m=2$), introduced by Johnson\cite{Johnson1972DM} and
Levenshte{\u\i}n \cite{Levenshtein:1971}.

Multiply constant-weight codes have attracted recent
attention due to an application in the design of
certain {\em physically unclonable functions} (PUFs)
\cite{Cheeetal:MCWC, Cherifetal:2012, CDGKS}.
Introduced by Pappu {\em et al.} \cite{Pappuetal:2002}, PUFs provide
innovative low-cost authentication methods that are derived from
complex physical characteristics of electronic devices 
 and have recently become an attractive option to provide security in low
cost devices such as RFIDs and smart cards
\cite{Cherifetal:2012,Gassendetal:2002,Pappuetal:2002,SuhDevadas:2007}.
Reliability and implementation considerations on programmable
circuits for the design of Loop PUFs \cite{Cherifetal:2012} lead to
the study of multiply constant-weight codes.


If we arrange each codeword in an MCWC as an $m\times n$
array, then an MCWC 
can be regarded as a code over binary matrices, where each matrix
has constant row weight $w$.
Furthermore, if certain weight constraints are also satisfied by all columns,
 we obtain {\em two-dimensional weight-constrained codes} that are
 used in the design of optical storage in holographic
memory \cite{Ordentlichetal:2000}, and
have possible applications in
next generation memory technologies based on crossbar arrays of
resistive devices, such as memristors \cite{Ordentlichetal:2012}.
Recently, these codes were also studied by Chee {\em et al.}
\cite{Cheeetal:2013c}  in an application for power line
communications.

The theory of MCWC is at a rudimentary stage.
Chee {\em et al.} established certain preliminary upper and lower bounds
for possible sizes of MCWCs \cite{Cheeetal:MCWC}. In particular,
they extended techniques of Johnson \cite{Johnson1972DM} to derive
certain upper bounds and showed that these bounds are asymptotically
tight to a constant factor. However, the only nontrivial infinite
class of optimal codes was constructed from Reed-Solomon codes under
the conditions that $n/w\ge mw-1$ is a prime power and $d\ge
2m(w-1)+2$.


In this paper, we continue this study and provide constructions of
optimal or near-optimal MCWCs based on combinatorial techniques.
Our main contributions are as follows:
\begin{enumerate}
\item[(i)] determining completely the maximum sizes of MCWCs with total weight four and distance four;
\item[(ii)] constructing infinite families of optimal MCWCs with distance four and weight two or three in each part;
\item[(iii)] establishing that the Johnson type upper bounds are asymptotically tight for fixed weight and distances.
\end{enumerate}

Our paper is organized as follows.
 In Section~\ref{pre}, we give necessary definitions and notation, where
a more general concept of MCWCs with different lengths and
 weights in different parts is introduced.
Section~\ref{sweight} links the general MCWCs to generalized packing
designs, a class of combinatorial objects recently studied by Bailey
and Burgess \cite{BB2011DM}. By establishing the existence of
optimal generalized packing designs, we completely determine the
maximum sizes of MCWCs with total weight four and distance four.
Combining results in \cite{BB2011DM}, the maximum sizes of MCWCs
with total weight less than or equal to four are determined except
for only one small case.
 Section~\ref{d4} gives a product construction of
 MCWCs with equal length and equal weight in each part. Based on the existence of
  large sets of optimal packings, this construction yields infinite families of optimal MCWCs with distance four
 and weight two or three in each part.
Section~\ref{asysec} exhibits that the Johnson type upper bounds of
general MCWCs are asymptotically tight for fixed  weights and
distance by applying a theorem on fractional matchings due to Kahn
\cite{kahn:1996}. This improves the previous result in
\cite{Cheeetal:MCWC}
which states that the bounds are asymptotically
tight to a constant factor and for a smaller class of MCWCs.
Finally in Section~\ref{2dmcwc}, we
formally define the notion of two dimensional MCWCs and
provide several bounds and constructions.

\section{Preliminaries}\label{pre}
Given a positive integer $n$, the set $\{1,2,\ldots,n\}$ is denoted
by $I_n$, and the ring $\bz/n\bz$ is denoted by $\bz_n$.

Through out this paper, let $R=\{0,1\}$ and let $X$ be a finite set
of size $N$. Then $R^X$ denotes the set of all binary vectors of
length $N$, where each component of a vector $\su\in R^X$ is indexed
by an element of $X$, that is, $\su=(\su_x)_{x\in X}$, and $\su_x\in
R$ for each $x\in X$.

A {\em binary code of length $N$} is a set $\C\subseteq R^X$ for
some $X$ of size $N$. The elements of $\C$ are called {\em
codewords}. The {\em Hamming norm} of a vector $\su\in R^X$ is
defined as $\|\su\|=|\{x\in X: \su_x=1\}|$. The distance induced by
this norm is called the {\em Hamming distance}, denoted $d_H$, so
that $d_H(\su,\sv) = \|\su-\sv\|$, for $\su,\sv\in R^X$. A code $\C$
is said to have (minimum Hamming) {\em distance} $d$ if
$d_H(\su,\sv)\geq d$ for all distinct $\su,\sv\in \C$, denoted by
$(N,d)$ code. The largest size of an $(N,d)$ code is denoted by
$A(N,d)$.

Suppose that $X_1,X_2,\ldots,X_m$ is a partition of $X$ with
$|X_i|=n_i$, $i\in I_m$. Clearly, $N=n_1+\cdots +n_m$. Define
$\vX=(X_1,X_2,\ldots,X_m)$ and $\vn=(n_1,n_2,\ldots,n_m)$. For any
vector $\su\in R^X$, define the {\em support of $\su$ associated
with $\vX$} as ${\rm supp}(\su)_{\vX}=(U_1,U_2,\ldots,U_m)$, where
$U_i=\{x\in X_i : \su_x =1\}$. If the subscript $\vX$ is omitted,
then ${\rm supp}(\su)=\cup_{i\in I_m}U_i$ is the usual support set.

Let $\vx=(x_1,\ldots,x_m)$ and $\vy=(y_1,\ldots,y_m)$ be two
$m$-tuples of nonnegative integers. We write $\vx\leq \vy$ to mean
that $x_i\leq y_i$ for all $i\in I_m$.  For any integer $k$, use the
notation $\binom{X}{k}$ to denote the set of all $k$-subsets of $X$.
Suppose $\vw=(w_1,w_2,\ldots,w_m)$ is an $m$-tuple of integers such
that $\vw\leq \vn$. Let $W \triangleq w_1+\cdots+w_m$.
 Define
\[\binom{\vX}{\vw}=\binom{X_1}{w_1}\times \binom{X_2}{w_2} \times \cdots \times \binom{X_m}{w_m},\]
i.e., a member of $\binom{\vX}{\vw}$ is an $m$-tuple of sets, of
sizes $(w_1,w_2,\ldots,w_m)$. Let $\C\subseteq R^X$ be an $(N,d)$
code. If for each $\su\in \C$, ${\rm supp}(\su)_{\vX} \in
\binom{\vX}{\vw}$, then $\C$ is said to be of {\em multiply
constant-weight}, denoted by MCWC$(\vn,d,\vw)$. The number of
codewords in an MCWC$(\vn,d,\vw)$ is called the {\em size} of the
code. The maximum size of an MCWC$(\vn,d,\vw)$ is denoted
$T(\vn,d,\vw)$, and the
 MCWC$(\vn,d,\vw)$ achieving this size is said to be
{\em optimal}.

Specifically, when $m=1$, an MCWC$(\vn,d,\vw)$ is a {\em
constant-weight code}, denoted by CWC$(n,d,w)$ with $n=n_1$ and
$w=w_1$; when $m=2$, an MCWC$(\vn,d,\vw)$ is a {\em doubly
constant-weight code} \cite{Etzion:2008}. The largest size of a
CWC$(n,d,w)$ is denoted by $A(n,d,w)$. For general $m$, when $n_i=n$
and $w_i=w$ for all $i\in I_m$, $\C$ is denoted by MCWC$(m, n,d,w)$.
In this case, we use the notation $M(m,n,d,w)$ to denote the maximum
size of such a code.
 Observe that by definition,
\[M(1,n,d,w)=A(n,d,w), M(2,n,d,w)=T((n,n),d,(w,w)).\]
Moreover, the functions $A(n,d,w)$ and $T((n,n),d,(w,w))$ have been
well studied, see for example
\cite{Agrelletal:2000,BSJSS1990IEEETIT,Etzion:2008,Johnson1972DM,SHP2006EJC}.
For the general case $T(\vn,d,\vw)$, some lower and upper bounds
were studied in \cite{Cheeetal:MCWC}. The techniques of
 Johnson bound have been applied to get the following recursive bounds in \cite{Cheeetal:MCWC,CDGKS}.

 \vskip 10pt
\begin{proposition}
\label{rb} We have
\begin{enumerate}
\item For each $i\in
I_m$, $T(\vn,d,\vw)\leq \lfloor
\frac{n_i}{w_i}T(\vn',d,\vw')\rfloor$ where
$\vn'=(n_1,\ldots,n_{i-1},n_i-1,n_{i+1},\ldots,n_m)$ and
$\vw'=(w_1,\ldots,w_{i-1},w_i-1,w_{i+1},\ldots,w_m)$.
\item $M(m,n,d,w)\leq
\left\lfloor\frac{n^m}{w^m}M(m,n-1,d,w-1) \right\rfloor$.
\end{enumerate}

\end{proposition}
\vskip 10pt

Since all the codes we consider in this paper are binary, we can
assume that the distance $d$ is even and let $d\triangleq 2\delta$
for convenience. Trivial cases are as follows.

\begin{lemma}\label{spread} If $d>2 \sum_{i\in I_m} w_i$, then
$T(\vn,d,\vw)=1$; if  $d=2 \sum^m_{i=1} w_i$, then
$T(\vn,d,\vw)=\min_{i\in I_m} \lfloor\frac{n_i}{w_i} \rfloor$; if
$d=2$, then $T(\vn,d,\vw)=\prod_{i\in I_m} {n_i\choose w_i}$.
\end{lemma}

%
%

%
\vskip 10pt


\section{Optimal MCWCs with Small Weight}\label{sweight}
We construct optimal multiply constant-weight codes with total
weight four in this section. Our approach is based on combinatorial
design theory. First, we introduce necessary concepts and establish
connections to multiply constant-weight codes.

\subsection{Generalized packing designs}

A {\em set system} is a pair $(X,\B)$ such that $X$ is a finite set
of {\em points} and $\B$ is a set of subsets of $X$, called {\em
blocks}. The {\em order} of the set system is $|X|$, the number of
points. For  a set $K$ of  nonnegative integers, a set system
$(X,\B)$ is said to be $K$-{\em uniform} if $|B|\in K$ for all
$B\in\B$. When $K=\{k\}$, we simply write that the system is
k-uniform.

Let $N\geq t$ and $\lambda\geq 1$. A $t$-$(N,K,\lambda)$ {\em
packing} is a $K$-uniform set system $(X,\B)$ of order $N$, such
that every $t$-subset of $X$ occurs in at most $\lambda$ blocks in
$\B$. If $K=\{k\}$ and each pair occurs in exactly $\lambda$ blocks,
then we call it a {\em balanced incomplete block design}, and denote
it by BIBD$(N,k,\lambda)$. A BIBD$(N,3,1)$ is known as a {\em
Steiner triple systems} of order $N$, denoted by STS$(N)$.


In \cite{Cameron:2009}, Cameron introduced a new class of
combinatorial designs, which simultaneously generalizes various
well-known classes of designs, including $t$-designs, mutually
orthogonal Latin squares, orthogonal arrays and $1$-factorizations
of complete graphs. In a recent paper \cite{BB2011DM}, Bailey and
Burgess considered the analogue of Cameron's generalization for
packings, which are called {\em generalized packing designs}. To
define generalized packing designs, we require various pieces of
notation and terminology.

If $\vA=(A_1,\ldots,A_m)$ and $\vB=(B_1,\ldots,B_m)$  are $m$-tuples
of sets, we write $\vA\subseteq\vB$ to mean that $A_i\subseteq B_i$
for all $i\in I_m$, and say $\vA$ is contained in $\vB$.

Again let $\vX=(X_1,X_2,\ldots,X_m)$ and $\vn=(n_1,n_2,\ldots,n_m)$
as in Section~\ref{pre}. Assume that $\vk=(k_1,\ldots,k_m)$ is an
$m$-tuple of positive integers with sum $k$ such that $\vk\leq \vn$.
Let $\vt=(t_1,\ldots,t_m)$ be an $m$-tuple of non-negative integers.
We say $\vt$ is {\em $(\vk, t)$-admissible} if $\vt\leq \vk$ and
$\sum_{i\in I_m} t_i=t$. In a similar manner, if
$\vT=(T_1,\ldots,T_m)$ is an $m$-tuple of disjoint sets, we say that
$\vT$ is {\em $(\vn,\vk,t)$-admissible} if each $T_i$ is a
$t_i$-subset of $X_i$, where $(t_1,\ldots,t_m)$ is $(\vk,
t)$-admissible. Note that since $t_i$ is allowed to be zero, the
corresponding set $T_i$ is allowed to be empty.

A {\em $t$-$(\vn, \vk, \lambda)$ generalized packing design}, or
more succinctly a {\em generalized packing}, is a pair $(\vX,\CP)$,
such that $\CP$ is a family of elements of $\binom{\vX}{\vk}$,
called {\em blocks}, with the property that every
$\vT=(T_1,\ldots,T_m)$ which is $(\vn,\vk,t)$-admissible is
contained in at most $\lambda$ blocks in $\CP$. As with ordinary
packings, the {\em generalized packing number}
$D_{\lambda}(\vn,\vk,t)$  is the maximum possible number of blocks
in a $t$-$(\vn, \vk, \lambda)$ generalized packing. When $\vn$ and
$\vk$ have only one component, we simply write $D_{\lambda}(N,k,t)$,
which is the usual packing number. When $\lambda$ is omitted, we
mean $\lambda=1$. The equivalence between multiply constant-weight
codes and generalized packing designs is obvious.

\begin{proposition}
\label{packing2code} There exists an MCWC$(\vn,d,\vw)$ of size $b$
if and only if a $t$-$(\vn, \vw, 1)$ generalized packing of size $b$
exists, where $t=W-\delta+1$.
\end{proposition}

\begin{corollary}
\label{packingbound} $T(\vn,d,\vw)= D(\vn,\vw,t)$ where
$t=W-\delta+1$.
\end{corollary}

A few general upper bounds for $D(\vn,\vk,t)$ can be found in
\cite{BB2011DM}. The cases for $t=2$ and $k=3$ or $4$ were
completely determined except when $\vk=(2,2)$, $n_1,n_2$ are both
odd and $n_1\leq n_2\leq 2n_1-1$ \cite{BB2011DM}. When $t=3$ and
$k=3$, the designs are trivial.  In the following subsection,
we determine $D(\vn,\vk,3)$ for $k=4$ completely. Note that the
upper bounds for all cases discussed later could be easily obtained
by applying techniques of Johnson. 


\subsection{Determination of $D(\vn,\vk,3)$ for $k=4$} We split the problem
 into five cases.

\vspace{0.2cm} Case 1: $\vk=(4)$. \vspace{0.1cm}

Let $\vn=(n)$, then a $3$-$(\vn, \vk, 1)$ generalized packing design
is indeed a $3$-$(n, 4, 1)$ packing, for which the determination of
packing numbers $D(n,4,3)$ has been completed by Bao and Ji in
\cite{BJ2014}. Hence, we have the following result.

\begin{proposition}\label{pn4}
For any positive integer $n$,
\begin{align*}
D((n),(4),3)=\begin{cases}
\left\lfloor\frac{n}{4}\left\lfloor\frac{n-1}{3}\left\lfloor\frac{n-2}{2}\right\rfloor\right\rfloor\right\rfloor & n\not\equiv 0\pmod{6}, \\
\left\lfloor\frac{n}{4}\left\lfloor\frac{n-1}{3}\left\lfloor\frac{n-2}{2}\right\rfloor\right\rfloor-1\right\rfloor & n\equiv 0\pmod{6}. \\
\end{cases}
\end{align*}
\end{proposition}

\vspace{0.2cm} Case 2: $\vk=(3,1)$. \vspace{0.1cm}

 Let
$\vn=(n_1,n_2)$, then we have $D(\vn,\vk,3)\leq \min\{{n_1\choose
3},\lfloor n_2D(n_1,3,2)\rfloor\}$. We prove that the upper bound is
achievable by using disjoint partial triple systems. Two
$2$-$(N,k,1)$ packings $(X,\A)$ and $(X,\B)$ are  {\em disjoint} if
$\A\cap\B=\emptyset$.

 When
$n\equiv 1,3\pmod{6}$, $n\neq 7$, there exists a {\em  large set} of
Steiner triple systems of order $n$, i.e., a set of $n-2$ pairwise disjoint
optimal $2$-$(n,3,1)$ packings,
\cite{Lu:1983,Lu:1984,Teirlinck:1991}. By collecting from each
STS$(n)$ the blocks not containing a fixed point, we obtain a set of
$n-2$ disjoint optimal $2$-$(n,3,1)$ packings of order $n-1$. For
$n\equiv 4\pmod{6}$, it was proved by Etzion
\cite{Etzion:1992a,Etzion:1992b} that there exists a partition of
all triples into $n-1$ optimal $2$-$(n,3,1)$ packings and one
$2$-$(n,3,1)$ packing of size $(n-1)/3$. For $n\equiv 5\pmod{6}$,
$n\geq 11$, Ji \cite{Ji:2006} proved that there is a partition of
triples into $n-2$ optimal $2$-$(n,3,1)$ packings and one packing of
size $4(n-2)/3$. Let $M(n)$ denote the maximum number of disjoint
optimal $2$-$(n,3,1)$ packings. Then for $n\geq 8$, $M(n)=n-2$ when
$n$ is odd and $M(n)=n-1$ when $n$ is even.

\begin{proposition}
$D((n_1,n_2),(3,1),3)=\min\{{n_1\choose 3},\lfloor
n_2D(n_1,3,2)\rfloor\}$ for all $n_1$, $n_2>0$ except for
$(n_1,n_2)\in \{(6,5),(7,3),(7,4),(7,5)\}$, whose values are listed
below.
\[\begin{array}{|c|c|c|c|c|}
\hline (n_1,n_2) &(6,5)&(7,3)&(7,4)&(7,5)\\ \hline D((n_1,n_2),(3,1),3) &18&20&26&31\\
\hline
\end{array}\]

\end{proposition}

\begin{proof} Let $\vX=(X_1,X_2)$ with $|X_i|=n_i$, $i=1,2$. Assume
that $(X_1,\B_i)$, $i\in I_{M(n_1)+\delta}$ are disjoint
$2$-$(n_1,3,1)$ packings as above with the first $M(n_1)$ being
optimal. Here $\delta=1$ when $n\equiv 4,5\pmod{6}$ and $0$
otherwise. For each $j\leq \min\{M(n_1)+\delta,n_2\}$, define
\[\C(j)=\{(B, \{x\}):B\in \B_i \text{ and } x \text{ is $i$-th
element of $X_2$}, 1\leq i \leq j\}.\]

For $n_1\geq 8$, if $n_2\leq M(n_1)$, then $(\vX,\C(n_2))$ is an
optimal generalized packing of size $\lfloor n_2D(n_1,3,2)\rfloor$;
if $n_2>M(n_1)$, then $(\vX,\C(M(n_1+\delta)))$ is optimal  with
${n_1\choose 3}$ blocks.

For $n_1\leq 4$, the optimal packings are trivial. For $n_1=5$, the
triples in $X_1$ can be partitioned into five disjoint optimal
$2$-$(5,3,1)$ packings, i.e., $M(5)=5$. For example, let
$\B_1=\{125,345\}$, $\B_2=\{135,234\}$, $\B_3=\{145,123\}$,
$\B_4=\{235,124\}$ and $\B_5=\{245,134\}$, then $(I_5, \B_i)$, $i\in
I_5$ are disjoint optimal $2$-$(5,3,1)$ packings. Then by the same
construction, we have $D(\vn,\vk,3)=2n_2$ if $n_2\leq 5$ and
$D(\vn,\vk,3)={5\choose 3}=10$ if $n_2>5$.

For $n_1=7$, Cayley \cite{Cayley:1850} showed that there are only
two mutually disjoint STS$(7)$s, i.e., $M(7)=2$. Now consider the
following six disjoint $2$-$(7,3,1)$ packings $(I_7,\B_i)$ with
$i\in I_6$:
\[\begin{array}{l}
\B_1=\{123,145,167,246,257,347,356\}, \\
\B_2=\{124,137,156,235,267,346,457\}, \\
\B_3=\{146,157,247,256,345,367\} ,\\
\B_4=\{125,136,147,234,357,456\}, \\
\B_5=\{126,135,237,245,567\},\\
\B_6=\{127,134,236,467\}.
\end{array}\]
So for $n_2\leq 2$, $(\vX,\C(n_2))$ is an optimal generalized
packing of size $\lfloor n_2D(n_1,3,2)\rfloor$; for $n_2=3$ or $4$,
$(\vX,\C(n_2))$ is an optimal generalized packing of size $\lfloor
n_2D(n_1,3,2)\rfloor-n_2+2$; for $n_2=5$, $(\vX,\C(5))$ is a
 generalized packing of size $31$, which is also optimal by an exhaustive computer search;
for all $n_2\geq 6$, $(\vX,\C(6))$ is an optimal generalized packing
of size ${n_1 \choose 3}$.

For $n_1=6$, delete all blocks containing the element $2$ from the
above six $2$-$(7,3,1)$ packings over $I_7$, then we get  six
disjoint $2$-$(6,3,1)$ packings over $I_7\setminus\{2\}$ with blocks
set $\B'_i$, where $|\B'_i|=4$ for $1\leq i\leq 4$ and $|\B'_i|=2$
for $i=5,6$. So for $n_2\leq 4$, $(\vX,\C(n_2))$ is an optimal
generalized packing of size $\lfloor n_2D(n_1,3,2)\rfloor$;  for
$n_2=5$, $(\vX,\C(5))$ is optimal of size $18$ by an exhaustive computer search;
for all $n_2\geq 6$, $(\vX,\C(6))$ is an optimal generalized packing
of size ${n_1 \choose 3}$.
\end{proof}

\vspace{0.2cm} Case 3: $\vk=(2,2)$.

\vspace{0.1cm} Let $\vn=(n_1,n_2)$ and assume $n_1\leq n_2$ without
loss of generality. Then the upper bound is
$D(\vn,\vk,3)\leq\frac{n_1(n_1-1)}{2}\left\lfloor\frac{n_2}{2}\right\rfloor$.

Denote $K_{n}$ a complete graph with $n$ vertices. Let
$\F(n)=\{F_1,F_2,\dots,F_{\gamma(n)}\}$ be a {\em $1$-factorization}
of $K_{n}$ if $n\equiv 0\pmod{2}$, or a {\em near $1$-factorization}
of $K_{n}$ if $n\equiv 1\pmod{2}$. Then $\gamma(n)=n-1$ if $n\equiv
0\pmod{2}$, and $\gamma(n)=n$  if $n\equiv 1\pmod{2}$. Further, each
{\em $1$-factor}  $F_i$ is of size
$\left\lfloor\frac{n}{2}\right\rfloor$, $i\in I_{\gamma(n)}$.

\begin{proposition}
$D((n_1,n_2),(2,2),3)=\frac{n_1(n_1-1)}{2}\left\lfloor\frac{n_2}{2}\right\rfloor$,
where $n_1\leq n_2$.
\end{proposition}

\begin{proof}
Let $\vX=(X_1,X_2)$ with $|X_i|=n_i$, $i=1,2$. Suppose that
$\F(n_i)$ is a (near) $1$-factorization of $K_{n_i}$ with vertex set
$X_i$, $i=1,2$. The upper bound is achieved by taking all blocks in
$\{(P,Q):P\in F_i\in\F(n_1),Q\in F_i'\in\F(n_2), 1\leq i\leq
\gamma(n_1)\}$.
\end{proof}

\vspace{0.2cm} Case 4: $\vk=(2,1,1)$. \vspace{0.1cm}

 Let
$\vn=(n_1,n_2,n_3)$ and assume $n_2\leq n_3$ without loss of
generality. Then the upper bound is $D(\vn,\vk,3)\leq
\min\{\frac{n_1(n_1-1)n_2}{2},\left\lfloor\frac{n_1}{2}\right\rfloor
n_2n_3\}$.

 An $n_2\times n_3$ {\em Latin rectangle} ($LR(n_2,n_3)$) is an $n_2\times n_3$ array over $I_{n_3}$ such
that each symbol occurs exactly once in each row, and at most once
in each column. When $n_2=n_3$, it is also called a {\em Latin
square} of order $n_3$. An $LR(n_2,n_3)$ can always be obtained from
a Latin square of order $n_3$ by collecting any $n_2$ rows.

\begin{proposition}
Let $n_2\leq n_3$. Then \[D((n_1,n_2,n_3),(2,1,1),3)=
\min\{\frac{n_1(n_1-1)n_2}{2},\left\lfloor\frac{n_1}{2}\right\rfloor
n_2n_3\}.\]
\end{proposition}

\begin{proof}
Let $\vX=(I_{n_1},I_{n_2},I_{n_3})$. Suppose that $M=(m_{ij})$ is an
$LR(n_2,n_3)$ over $I_{n_3}$ and $\F(n_1)$ is a (near)
$1$-factorization of $K_{n_1}$ with vertex set $I_{n_1}$. For each
$s\leq \min\{\gamma(n_1), n_3\}$, define
\[\C(s)=\{(P,\{i\},\{j\}):P\in
F_{x}\in\F(n_1) \text{ and }   i\in I_{n_2},j\in I_{n_3} \text{ with
} m_{ij}=x, 1\leq x\leq s \}.\] Then when $\gamma(n_1)\geq n_3$,
$(\vX,\C(n_3))$ is an optimal
generalized packing of size
$\left\lfloor\frac{n_1}{2}\right\rfloor n_2n_3$; when
$\gamma(n_1)<n_3$, $(\vX,\C(\gamma(n_1)))$ is an optimal
generalized packing of size $\frac{n_1(n_1-1)n_2}{2}$.
\end{proof}

\vspace{0.2cm} Case 5: $\vk=(1,1,1,1)$. \vspace{0.1cm}

\begin{proposition}
Let $\vn=(n_1,n_2,n_3,n_4)$ and $n_1\leq n_2\leq n_3\leq n_4$. Then
$D(\vn,(1,1,1,1),3)=n_1n_2n_3$.
\end{proposition}

\begin{proof}
It is clear that $D(\vn,\vk,3)\leq n_1n_2n_3$.
Consider $\bz_{n_4}$ and let $X_i=\{s_i: 0\le s_i\le n_i-1\}$ for $i\in I_4$.
Equality then follows from considering all blocks of the form
$(\{s_1\},\{s_2\},\{s_3\},\{s_4\})$, where $s_i\in X_i$ and $\sum_{i\in I_4}s_i=0$.
Here the addition is over $\bz_{n_4}$.
 \end{proof}

Combining results in \cite{BB2011DM}, we have determined
$T(\vn,d,\vw)$ for total weight less than or equal to four
completely except when $\vw=(2,2)$, $d=6$ and $n_1\leq n_2\leq
2n_1-1$, both $n_1$ and $n_2$ are odd.

\section{Optimal MCWCs with Minimum Distance Four}\label{d4}
In this section we handle the multiply constant-weight codes
MCWC$(m, n,d,w)$ with minimum distance four and small weight $w$.
For convenience, we sometimes arrange a codeword as an $m\times n$
binary matrix.


When $w = 1$, an MCWC$(m,n,d,1)$ is equivalent to an $n$-ary code of
length $m$ and distance $\delta$ by a bijection from row words to
the set $\{0,1,\ldots,n-1\}$.

\begin{lemma}[Chee {\em et al.} \cite{Cheeetal:MCWC}] \label{w=1}
 $M(m,n,d,1)=A_n(m,\delta)$, where $A_n(m,\delta)$ is the maximum size of an $n$-ary code of length $m$ and distance $\delta$.
\end{lemma}

When $\delta=2$, $A_n(m,2)=n^{m-1}$ which can be achieved by a code
over $\bz_n$ consisting of all $m$-tuples with a constant sum.

\begin{corollary}\label{w=1d=4}
 $M(m,n,4,1)= n^{m-1}$.
\end{corollary}

%
%

Applying Proposition~\ref{rb}(i) iteratively  $w$ times for each
$i\in I_{m-1}$, we obtain the following consequence.
\begin{proposition}\label{sw}
If $d\leq 2w$, then $M(m,n,d,w)\leq \binom{n}{w}^{m-1}A(n,d,w)$.
\end{proposition}

By Corollary~\ref{w=1d=4}, the bound in Proposition~\ref{sw} is
tight when $w=1$ and $d=4$. Next, we show that it is also tight for
some other small values $w$. The following two corollaries are immediate
consequences of Proposition~\ref{sw}.

\begin{corollary}\label{sw=2}
 $M(m,n,4,2)\leq \binom{n}{2}^{m-1}\lfloor \frac{n}{2}\rfloor$.
\end{corollary}

\begin{corollary}\label{sw=3}
 $M(m,n,4,3)\leq \binom{n}{3}^{m-1} \lfloor \frac{n}{3} \lfloor\frac{n-1}{2}
 \rfloor\rfloor$.
\end{corollary}

\subsection{Product Construction}
We give a product construction as follows by generalizing the method
used in \cite[Section 6]{Etzion:2008}.

 Let $Y$ be a finite set of size $n$. Suppose that
$\C_i\subseteq R^Y$ is a CWC$(n,d_1,w)$ for each $i\in I_s$ and
$\cup_{i\in I_s} \C_i$ is a CWC$(n,d_2,w)$. Obviously, we have
$d_2\leq d_1$. Further, let $\C_0\subset R^{I_m\times I_s}$ be an
MCWC$(m,s,d_3,1)$. Let $X=I_m\times Y$ and $\vX=(\{1\}\times
Y,\ldots,\{m\}\times Y)$. For each codeword $\su\in \C_0$, define a
code $\D_{\su}\subset R^{X}$ as
\[\D_{\su}=\{ \sv^{j_1}| \sv^{j_2}|\cdots| \sv^{j_m}: \sv^{j_i}\in\C_{j_i},  (i,j_i)\in {\rm
supp}(\su) \text{ and } i\in I_m\},\] where the $|$ operator
concatenates codewords as strings.

Then for each $c\in \D_{\su}$, ${\rm supp}(c)_{\vX}\in
\binom{\vX}{\vw}$, where $\vw=(w,\ldots,w)$. Let $\D=\cup_{\su\in
\C_0}\D_{\su}$. It is easy to check that $\D$ is an MCWC$(m,n,d,w)$
with minimum distance $d\geq \min(d_3d_2/2,d_1)$. In fact, for any
two codewords in $\D_{\su}$, the distance is at least $d_1$. For two
codewords $c\in \D_{\su}$ and $c'\in \D_{\sv}$ with $\su\neq \sv$,
since $d_H(\su,\sv)\geq d_3$, $|{\rm supp}(\su)\setminus {\rm
supp}(\sv)|\geq d_3/2$. Hence $d_H(c,c')\geq d_3d_2/2$. Further, if
all $\C_i$ have the same size $M$, then  $|\D|=|\C_0|M^m$.

\subsection{Optimal MCWCs}

Now we apply the product construction with special families of
constant-weight codes to obtain optimal multiply constant-weight
codes.

\begin{proposition}\label{sw=2d=4}
 $M(m,n,4,2)= \binom{n}{2}^{m-1}\lfloor \frac{n}{2}\rfloor$ for all
 positive integers
 $n$ and $m$.
\end{proposition}
\begin{proof} When $n$ is even, let $s=n-1$.  For each $i\in I_s$, let $\C_i\subset R^Y$
be a CWC$(n,4,2)$ such that $F_i=\{{\rm supp}(\su):\su\in \C_i\}$ is
a $1$-factor and $\{F_i:i\in I_s\}$ is a $1$-factorization of $K_n$
with vertex set $Y$. Then apply the product construction  with
$d_1=d_3=4$, $d_2=2$ and $w=2$.

When $n$ is odd, let $s=n$. Then use the similar codes such that
each $F_i$ is a near $1$-factor and $\{F_i:i\in I_s\}$ is a near
$1$-factorization.
\end{proof}

To get the following result, we again use the existence of
partitions of triples into disjoint optimal packings as described in
Section \ref{sweight}.

\begin{proposition}\label{sw=3d=4}
 $M(m,n,4,3)=  \binom{n}{3}^{m-1}   \lfloor \frac{n}{3} \lfloor\frac{n-1}{2}
 \rfloor\rfloor$ for all positive integers $m$ and $n\equiv 0,1,2,3\pmod 6$ with $n\neq 6,7$.
\end{proposition}
\begin{proof} For each $n\equiv 1,3\pmod 6$ and $n\neq 7$, let $s=n-2$. For each $i\in I_s$, let
$\C_i\subset R^Y$ be a CWC$(n,4,3)$ such that $F_i=\{{\rm
supp}(\su):\su\in \C_i\}$ is the block set of a Steiner triple
system over $Y$ and $\{F_i:i\in I_s\}$ forms a large set of Steiner
triple systems. Then apply the product construction  with
$d_1=d_3=4$, $d_2=2$ and $w=3$.

For each $n\equiv 0,2\pmod 6$ and $n\neq 6$, let $s=n-1$. For each
$i\in I_s$, let $\C_i\subset R^Y$ be a CWC$(n,4,3)$ such that
$F_i=\{{\rm supp}(\su):\su\in \C_i\}$ is an optimal $2$-$(n,3,1)$
packing over $Y$ of size $n(n-2)/6$ and $\{F_i:i\in I_s\}$ is a set
of $n-1$ disjoint optimal $2$-$(n,3,1)$ packings. Then apply the
product construction  with same values of $d_1,d_2,d_3$ and $w$.
\end{proof}

Note that when $n\equiv 4,5\pmod 6$, there is no partition of all
triples into optimal packings
\cite{Etzion:1992a,Etzion:1992b,Ji:2006}, but only a partition
 with almost all packings except one being optimal.
By the product construction, if not all $\C_i$ have the same size,
then it is difficult to compute the exact size of $\D$. So we can
only get a lower bound of $\D$ by using disjoint optimal packings
for these cases.

\begin{proposition}\label{lbsw=3d=4}
Let $n\geq 10$. For  all positive integers $m$, $M(m,n,4,3)\geq
(\frac{n(n-2)-2}{6})^m(n-1)^{m-1}$ when $n\equiv 4\pmod 6$ and
$M(m,n,4,3)\geq (\frac{n(n-1)-8}{6})^m(n-2)^{m-1}$ when $n\equiv
5\pmod 6$.
\end{proposition}

\section{Asymptotic sizes of $T(\vn,d,\vw)$}\label{asysec}
The following result for $M(m,n,d,w)$ was proved in
\cite{Cheeetal:MCWC}.

\begin{proposition}[Chee {\em et al.} \cite{Cheeetal:MCWC}]
\label{dgeq2m} Given $m,d,w$, let $s$ be the
smallest integer such that $m(w-s)-\delta+1<m$ and
$r=m(w-s)-\delta+1$. Then
\begin{equation}\label{eq:Mmndw}M(m,n,d,w)\leq \left\lfloor\frac{n^m}{w^m}\left\lfloor\frac{(n-1)^m}{(w-1)^m}
\cdots \left\lfloor\frac{(n-s+1)^m}{(w-s+1)^m} \left\lfloor
\frac{(n-s)^r}{(w-s)^r}\right\rfloor \right\rfloor \cdots
\right\rfloor\right\rfloor.\end{equation} Further, let
$t=mw-\delta+1$ and consider $M(m,n,d,w)$ as a function of $n$, then
\begin{equation*}
1\le \limsup_{n\to\infty} \frac{M(m,n,d,w)}{n^t/w^t} \le
\frac{w^t}{(w-s)^t}.
\end{equation*}
In addition, when $t\le m$, $n/w\ge mw-1$ and $n/w$ is a prime
power,  $M(m,n,d,w)=\frac {n^t}{w^t}$ holds.
\end{proposition}
\vspace{0.5cm}

The last statement of Proposition \ref{dgeq2m} shows that under
certain restrictions, the upper estimate (\ref{eq:Mmndw}) is
asymptotically sharp. In this section, we prove that
(\ref{eq:Mmndw}) is asymptotically sharp for all cases. In fact, we
give a more general result for $T(\vn,d,\vw)$ when all the
components of $\vn$ grow in any fixed proportion.

\subsection{Asymptotic theorem}
Since $T(\vn,d,\vw)= D(\vn,\vw,t)$, where $t=W-\delta+1$ by
Corollary ~\ref{packingbound}, we study the upper bound of
generalized packing numbers $D(\vn,\vw,t)$ instead.

 Let $\vX,\vn,\vw,t$ be defined as before. Suppose that $n_i=c_iv$, $i\in I_m$ for some integer $v$. Given a
$(\vw, t)$-admissible $\vt$, since every element of ${\vX \choose
\vt}$ occurs at most once in a block, we have
\begin{equation}\label{upb}D(\vn,\vw,t)\leq \prod_{i\in
I_m}\binom{c_iv}{t_i}/\binom{w_i}{t_i}\leq \frac{v^t}{\prod_{i\in
I_m}w_i!}\prod_{i\in I_m}c_i^{t_i}(w_i-t_i)!.\end{equation}

Let $C\triangleq \min\{\prod_{i\in I_m}c_i^{t_i}(w_i-t_i)!: \vt
\text{ is $(\vw, t)$-admissible} \}$. Then our asymptotic result can
be stated as follows.

\begin{theorem}\label{asythm} Let $t=W-\delta+1$ and $d=2\delta$.
Then
\begin{equation*}\lim_{v\to\infty}\frac{T(\vn,d,\vw)}{v^t}=\lim_{v\to\infty}\frac{D(\vn,\vw,t)}{v^t}=\frac{C}{\prod_{i\in
I_m}w_i!}.
\end{equation*}
\end{theorem}
%
%

For MCWC$(m, n,d,w)$, we can let $c_i=1$ for all $i\in I_m$ and
$v=n$. Then $C$ can be achieved when the $m$-tuple $\vt$ is an
almost constant tuple, that is, values of $t_i$ differ at most one.
Applying Theorem~\ref{asythm}, we get the asymptotic sizes for
$M(m,n,d,w)$.

\begin{corollary}\label{asy} Given $m,d,w$, let $s$ be the
smallest integer such that $m(w-s)-\delta+1<m$ and
$r=m(w-s)-\delta+1$. Let $t=mw-\delta+1$. Then
\begin{equation*}
\lim_{n\to\infty} \frac{M(m,n,d,w)}{n^t}= \frac{1}{w^m(w-1)^m\cdots
(w-s+1)^m (w-s)^r }.
\end{equation*}
\end{corollary}

\subsection{Proof of Theorem~\ref{asythm}}
Given a  hypergraph $\CH$, let $ E(\CH)$ be the edge set and
$V(\CH)$ be the vertex set. Denote $\nu(\CH)$ the maximum size of a
matching in $\CH$.

 A function $\theta: E(\CH)\longrightarrow R$ is a {\em fractional
 matching} of the hypergraph $\CH$ if $\sum_{e\in E(\CH);x\in e} \theta(e)\leq
 1$ holds for every vertex $x\in V(\CH)$. Let $\theta(\CH)=\sum_{e\in E(\CH)}
 \theta(e)$. The {\em fractional matching
 number}, denoted $\nu^*(\CH)$ is the maximum of $\theta(\CH)$ over all fractional matchings.
Clearly,
\[\nu(\CH)\leq \nu^*(\CH).\]
Kahn \cite{kahn:1996} proved that under certain conditions,
asymptotic equality holds. For a subset $W\subset V(\CH)$, define
$\bar{\theta}(W)=\sum_{ e\in E(\H);W\subset e}\theta(e)$ and $\alpha
(\theta)=\max\{\bar{\theta}(\{x,y\}):x,y\in V(\CH), x\neq y\}$. Here
$\alpha (\theta)$ is a fractional generalization of the codegree. We
say that $\CH$ is {\em $l$-bounded} if each of its edges has size at
most $l$.
\begin{theorem}[Kahn \cite{kahn:1996}]\label{kahn}
For every $l$ and every $\varepsilon>0$ there is a $\sigma$ such
that whenever $\CH$ is an $l$-bounded hypergraph and $\theta$ is a
fractional matching with $\alpha(\theta)<\sigma$, then
\[\nu(\CH)>(1-\varepsilon)\theta(\CH).\]
\end{theorem}

We apply Kahn's Theorem to prove Theorem~\ref{asythm}. First, we
define a hypergraph $\Gamma_v$ with $v$ defined in
$\vn=(c_1v,\ldots,c_mv)$.

Let $\CT$ be the set of all $(\vw, t)$-admissible $m$-tuples. Given
$\vA \in \binom{\vX}{\vw}$, let $\E(\vA)=\cup_{\vt\in \CT}
\binom{\vA}{\vt}$. Then construct a hypergraph $\Gamma_v$ with
vertex set $\cup_{\vt\in \CT} \binom{\vX}{\vt}$ and edge set
$\{\E(\vA):\vA \in \binom{\vX}{\vw}\}$. Note that for any two
distinct blocks $\vA_1$ and $\vA_2$ in the generalized packing, we
have $\E(\vA_1) \cap \E(\vA_2)=\emptyset$. Hence a $t$-$(\vn,\vw,1)$
generalized packing corresponds to a matching
 in $\Gamma_v$, i.e., $\nu(\Gamma_v)=D(\vn,\vw,t)$.

It suffices to verify the conditions of Theorem~\ref{kahn} and to
produce a fractional matching $\theta$ of the hypergraph $\Gamma_v$
of the desired size. It is easy to know that $\Gamma_v$ is
$l$-bounded with $l=\sum_{\vt\in \CT}\prod_{i\in I_m}\binom{w_i}{t_i}$. Now consider
the function $\theta: E(\Gamma_v)\longrightarrow R$ by
\[\theta(e)=\frac{C}{v^{W-t}\prod_{i\in I_m}c_i^{w_i}}\ , \]
for every $e\in E(\Gamma_v)$. We first check $\theta$ is a
fractional matching. For any vertex $x\in V(\Gamma_v)$, which is an
$(\vn,\vw,t)$-admissible $m$-tuple of disjoint sets of sizes
 $\vt$, we have

\begin{equation*}
\begin{split}
\deg(x)=&\prod_{i\in I_m} \binom{c_iv-t_i}{w_i-t_i}\leq \prod_{i\in
I_m} \frac{(c_iv)^{w_i-t_i}}{(w_i-t_i)!}\\=&v^{W-t}\prod_{i\in I_m}
\frac{c_i^{w_i}}{c_i^{t_i}(w_i-t_i)!}\leq \frac{v^{W-t}\prod_{i\in
I_m} c_i^{w_i}}{C}.
\end{split}
\end{equation*}
Hence, $\sum_{e\in E(\Gamma_v);x\in e} \theta(e)\leq 1$ and $\theta$
is indeed a fractional matching. Next, we compute  $\alpha(\theta)$.
For every $x,y\in V(\Gamma_v)$, then $x\cup y$ is
$(\vn,\vw,t')$-admissible with $t'\geq t+1$. Here, the union
operation is component-wise. Then the codegree of $x$ and $y$ is
\[\deg(x,y)=O(v^{W-t-1}).\]
Hence $\alpha(\theta)=\deg(x,y)\cdot \theta(e)=o(1)$ when
$v\to\infty$.

Finally, we apply Kahn's Theorem.
\begin{equation*}
\begin{split}
\lim_{v\to\infty} \frac{D(\vn,\vw,t)}{v^t} & =\lim_{v\to\infty}
\frac{\nu(\Gamma_v)}{v^t}  \geq \lim_{v\to\infty}
\frac{\theta(\Gamma_v)}{v^t}
= \lim_{v\to\infty} \frac{|E(\Gamma_v)|\times \theta(e)}{v^t}\\
&= \lim_{v\to\infty}  \prod_{i\in I_m}\binom{c_iv}{w_i} \times
\frac{C}{v^{W-t}\prod_{i\in I_m}c_i^{w_i}}/v^t=\frac{C}{\prod_{i\in
I_m}w_i!}.
\end{split}
\end{equation*}
The other inequality comes from the upper bound (\ref{upb}).

\section{Two Dimensional Multiply Constant-Weight
Codes}\label{2dmcwc} Recall that when the lengths of different parts
of a codeword are constant, say $\vn=(n,n,\ldots,n)$, then each
codeword could be considered as an $m\times n$ binary matrix.
In this section, we impose additional weight constraints on all columns.
We note that these codes have applications
in optical storage in holographic memory \cite{Ordentlichetal:2000},
 crossbar arrays of resistive devices \cite{Ordentlichetal:2012}, and
power line communications \cite{Cheeetal:2013c}.

Let $\vn=(n,n,\ldots,n)$ and $\vw=(w_1,\ldots,w_m)$. A {\em two
dimensional multiply constant-weight code} $2$DMCWC$(m,n,d,\vw,l)$
is an MCWC$(\vn,d,\vw)$ in a matrix form such that each column of
 codewords has constant weight $l$. Let $M(m,n,d,\vw,l)$ denote
the largest size of a $2$DMCWC$(m,n,d,\vw,l)$. If $w_i=w$ for all
$i\in I_m$, we simply write $2$DMCWC$(m,n,d,w,l)$ and
$M(m,n,d,w,l)$.

\subsection{Upper Bounds}

%

An {\em $\alpha$-parallel class} of a set system is a subset of the blocks that each point appears exactly $\alpha$ times. 

\begin{definition}
Let $(Y,\B)$ be a $2$-$(M,K,\lambda)$ packing. If the blocks can be
arranged into an $m \times n$ array $R$ such that
\begin{enumerate}
\item[(1)] each entry in $R$ is either empty or a block;
\item[(2)] the blocks in the $i$-th row form a $w_i$-parallel class;
\item[(3)] the blocks in each column form an $l$-parallel class.
\end{enumerate}
Then, we call it a {\em doubly resolvable packing}, and denote it by
DRP$(M,\lambda;\vw,l;m,n)$. Again we write DRP$(M,\lambda;w,l;m,n)$
if $w_i=w$ for all $i\in I_m$.
\end{definition}

\begin{proposition}
A DRP$(M,\lambda;\vw,l;m,n)$ is equivalent to a $2$DMCWC$(m,n,d,$
$\vw,l)$ of size $M$, where $d=2(W-\lambda)$.
\end{proposition}

\begin{proof}
Suppose $(Y,\B)$ is a DRP$(M,\lambda;\vw,l;m,n)$ such that the
blocks are arranged into an $m\times n$ array $R$. Denote the
$(i,j)$-th entry of $R$ by $R_{ij}$, where $i\in I_m$ and $j\in
I_n$.

For each $x\in Y$, we construct an $m\times n$ binary matrix $\su^x$
with the $(i,j)$-th entry defined as
\[\su^x_{ij} =
\begin{cases}
1, & \text{if } x \in R_{ij}; \\
0, & \text{otherwise.}
\end{cases}\]
Then $\C=\{\su^x\mid x\in Y\}$ is  a $2$DMCWC$(m,n,d,\vw,l)$ of size
$M$, where $d=2(W-\lambda)$.

The construction can be easily reversed to obtain the converse.
\end{proof}

\begin{example}
\label{small} Here is an example of a DRP$(3,3;2,2;3,3)$ over
$Y=\bz_3$,

\[
\begin{array}{|c|c|c|}
\hline
01 & 12 & 02 \\
\hline
02 & 01 & 12 \\
\hline
12 & 02 & 01 \\
\hline
\end{array}\]
from which we obtain a $2$DMCWC$(3,3,6,2,2)$ by taking the codewords

\[\su^0=\left(\begin{array}{lll}
1 & 0 & 1 \\
1 & 1 & 0 \\
0 & 1 & 1 \\
\end{array}\right),
\su^1=\left(\begin{array}{lll}
1 & 1 & 0 \\
0 & 1 & 1 \\
1 & 0 & 1 \\
\end{array}\right),
\su^2=\left(\begin{array}{lll}
0 & 1 & 1 \\
1 & 0 & 1 \\
1 & 1 & 0 \\
\end{array}\right).\]
\end{example}

\begin{lemma}[Upper Bound]
\label{bound} If $\sum_{i\in I_m} w_i^2-n\lambda>0$, and there
exists a DRP$(M,\lambda;$ $\vw,l;m,n)$, or equivalently a
$2$DMCWC$(m,n,d,\vw,l)$ of size $M$, where $d=2(W-\lambda)$, then
$$M\le\dfrac{n(nl-\lambda)}{\sum_{i\in I_m} w_i^2-n\lambda}.$$
\end{lemma}

\begin{proof}
Let $R$ be the $m\times n$ array of a DRP$(M,\lambda;\vw,l;m,n)$ and
$f_{ij}=|R_{ij}|$. Then
\[\sum_{i\in I_m}\sum_{j\in I_n}f_{ij}=Mnl=M\sum_{i\in I_m} w_i.\]
By the definition of a doubly resolvable packing,
\begin{eqnarray*}
\lambda M(M-1) & \ge & \sum_{i\in I_m} \sum_{j\in I_n} f_{ij}(f_{ij}-1) \\
& = & \sum_{i\in I_m}\sum_{j\in I_n} f_{ij}^2-Mnl.
\end{eqnarray*}
By Cauchy--Schwartz inequality, for any $i\in I_m$,
\begin{eqnarray*}
\sum_{j\in I_n} f_{ij}^2\ge\dfrac{(\sum_{j\in I_n}
f_{ij})^2}{n}=\dfrac{(Mw_i)^2}{n}.
\end{eqnarray*}
So
\[\lambda M(M-1)\ge \sum_{i\in  I_m}\dfrac{(Mw_i)^2}{n}-Mnl.\]
Hence
\[M\le\dfrac{n(nl-\lambda)}{\sum_{i\in I_m} w_i^2-n\lambda},\]
provided that $\sum_{i\in I_m} w_i^2-n\lambda>0$. Note that the
right hand side is always positive since $n^2l=n\sum_{i\in I_m}
w_i\geq \sum_{i\in I_m} w_i^2$.
\end{proof}

If the bound in Lemma~\ref{bound} is achieved, then all the blocks in
the $i$-th row are of the same size $\frac{Mw_i}{n}$, and each
pair appears in exactly $\lambda$ blocks. It is easy to check that
the $2$DMCWC$(3,3,6,2,2)$ constructed in Example~\ref{small} is
optimal.

Further, if we let $f=\left\lfloor\dfrac{Ml}{m}\right\rfloor$ and
$r=Ml-mf$. Then we can improve the above bound by using
$$\sum_{i\in I_m}f_{ij}^2\geq (m-r)f^2+r(f+1)^2.$$

\begin{lemma}[Improved Upper Bound] If there
exists a DRP$(M,\lambda;\vw,l;m,n)$,  or equivalently a
$2$DMCWC$($ $m,n,d,\vw,l)$ of size $M$, where $d=2(W-\lambda)$, then
\label{ibound}
\begin{eqnarray*}
\lambda M(M-1) & \ge & n[(m-r)f^2+r(f+1)^2]-Mnl \\
& = & n(mf^2+2rf+r)-Mnl,
\end{eqnarray*}where $f=\left\lfloor\dfrac{Ml}{m}\right\rfloor$ and
$r=Ml-mf$.
\end{lemma}

Note that the upper bound of $M(m,n,d,\vw,l)$ is the biggest $M$
satisfies the
 inequality in Lemma~\ref{ibound}, which is usually achieved when equality holds.

\begin{example}
\label{osmall}Here are two examples of optimal two dimensional
multiply constant-weight codes that achieve the bound in
Lemma~\ref{ibound}. We list the equivalent doubly resolvable
packings instead.

An optimal $2$DMCWC$(6,6,20,2,2)$ of size $4$:
$$\begin{array}{|c|c|c|c|c|c|}
\hline 01 & 23 & 0 & 1 & 2 & 3 \\ \hline 23 & 01 & 3 & 0 & 1 & 2 \\
\hline 2 & 3 & 02 & 13 & 0 & 1 \\ \hline 1 & 2 & 13 & 02 & 3 & 0 \\
\hline 0 & 1 & 2 & 3 & 03 & 12 \\ \hline 3 & 0 & 1 & 2 & 12 & 03 \\
\hline
\end{array}\ .$$

An optimal $2$DMCWC$(9,9,32,2,2)$ of size $6$:
$$\begin{array}{|c|c|c|c|c|c|c|c|c|}
\hline 01 & 45 & 12 & 2 & 3 & 0 & 3 & 5 & 4 \\ \hline 25 & 13 & 04 &
3 & 0 & 1 & 4 & 2 & 5 \\ \hline 34 & 02 & 35 & 0 & 1 & 2 & 5 & 4 & 1
\\ \hline
4 & 5 & 5 & 03 & 14 & 03 & 2 & 1 & 2 \\ \hline 3 & 4 & 2 & 15 & 05 &
24 & 1 & 3 & 0 \\ \hline 5 & 1 & 4 & 24 & 23 & 15 & 0 & 0 & 3 \\
\hline 0 & 2 & 3 & 5 & 4 & 3 & 01 & 45 & 12 \\ \hline 1 & 3 & 0 & 4
& 2 & 5 & 25 & 13 & 04 \\ \hline 2 & 0 & 1 & 1 & 5 & 4 & 34 & 02 &
35 \\ \hline
\end{array}\ .$$
\end{example}

\subsection{Constructions} First we show that concatenating small optimal
two dimensional MCWCs gives big optimal two dimensional MCWCs.

\begin{proposition}
Let $a$ be a positive integer. Suppose there exists an optimal
$2$DMCWC$(m,n,d,\vw,l)$ of size $M$ achieving the bound in
Lemma~\ref{bound} or Lemma~\ref{ibound}. Then there exists an
optimal $2$DMCWC$(am,n,$ $ad,\vw',al)$ of size $M$, where $\vw'$ is a
vector of length $am$ by copying $\vw$ $a$ times.
\end{proposition}

\begin{proof}First, we check that the two codes have the
same upper bounds of sizes. Suppose that a $2$DMCWC$(m,n,d,\vw,l)$
corresponds to a DRP$(M,\lambda;\vw,l;m,n)$ and a
$2$DMCWC$(am,n,ad,\vw',al)$ corresponds to a
DRP$(M_1,\lambda_1;\vw',al;am,n)$, then $\lambda_1=a\lambda$. By
Lemma~~\ref{bound}, $M$ and $M_1$ satisfy the same inequality. Now
we check for Lemma~\ref{ibound}. Let
$f_1=\left\lfloor\dfrac{M_1al}{am}\right\rfloor=\left\lfloor\dfrac{M_1l}{m}\right\rfloor$
and $r_1=M_1al-amf_1=a(M_1l-mf_1)$. Then

\[\lambda_1 M_1(M_1-1)\ge n[(am-r_1)f_1^2+r_1(f_1+1)^2]-aM_1nl,\]
that is,
\[\lambda M_1(M_1-1)\ge n[(m-r_1/a)f_1^2+(r_1/a)(f_1+1)^2]-M_1nl.\]
Being considered as an inequality with indeterminate $M_1$, it is
exactly the same inequality as in Lemma~\ref{ibound}. Hence, $M$ and
$M_1$ have the same restrictions.

So if there exists an optimal $2$DMCWC$(m,n,d,\vw,l)$ $\C$, then we
can obtain an optimal $2$DMCWC$(am,n,$ $ad,\vw',al)$ by concatenating
each codeword from $\C$ $a$ times in the  vertical direction.
\end{proof}

By applying the same technique but concatenating each codeword in
$\C$ $b$ times in the  horizontal direction, we have that
$M(m,n,d,\vw,l)=M$ achieving the bound in Lemma~\ref{bound} or
Lemma~\ref{ibound} implies $M(am,bn,abd,b\vw',al)=M$. Hence, we only
consider the codes with parameters $m$, $l$ and $d$ (or $n$, $\vw$
and $d$) having no common divisors.

\begin{lemma}
For any positive integer $n$, $M(n,n,2n,1,1)=n$.
\end{lemma}

\begin{proof}
A Latin square of order $n$ is a DRP$(n,0;1,1;n,n)$.
\end{proof}

\vskip 10pt

The construction of $2$DMCWC$(nl,n,d,1,l)$s has been
investigated in many papers as equidistant frequency permutation
arrays and constant-composition codes
\cite{DingYin:2006,Huczynska:2010}. Most of the constructions can be
generalized to construct $2$DMCWC$(m,n,d,$ $w,l)$.


\begin{construction}\label{abibd}
If there exists an $\alpha$-resolvable BIBD$(M,k,\lambda)$ with
$r=\frac{(M-1)\lambda}{(k-1)\alpha}$ $\alpha$-parallel classes (each
has $b=\frac{\alpha M}{k}$ blocks), then for any pair of positive
integers $(s,t)$ with $st=r$, there exists a DRP$(M,b\lambda;\alpha
t,\alpha s;bs,bt)$, i.e., an optimal  $2$DMCWC$(bs,bt,d,\alpha
t,\alpha s)$ of  distance $d=2b(\alpha r-\lambda)$.
\end{construction}
\begin{proof} We construct the $bs\times bt$ array of a DRP$(M,b\lambda;\alpha
t,\alpha s;bs,bt)$ as follows.  For each $1\leq i\leq r$, let $A_i$
be a $b\times b$ array with the first column occupied by the blocks
in the $i$-th $\alpha$-parallel class, and other columns being
cyclic shifts of the first one. Then the $bs\times bt$ array $R$ is
formed by arranging all $A_i$'s into an $s\times t$ array. The
optimality is easy to be checked by Lemma~\ref{bound}.
\end{proof}

By Construction~\ref{abibd}, the existence of  $\alpha$-resolvable
BIBD$(M,k,\lambda)$'s implies the existence of optimal two
dimensional multiply constant-weight codes with certain parameters.
Necessary conditions for the existence of an $\alpha$-resolvable
BIBD$(M,k,\lambda)$ are (1) $\lambda(M-1)\equiv
0\pmod{\alpha(k-1)}$, (2) $\lambda M(M-1)\equiv 0\pmod{k(k-1)}$, and
(3) $\alpha M\equiv 0\pmod{k}$. For $k\in \{2,3,4\}$, the necessary
conditions are sufficient except for $(M,k,\lambda,\alpha)\in
\{(6,3,\lambda,1):\lambda\equiv 2\pmod{4}\} \cup\{(10,4,2,2)\}$
 \cite{Jungnickel:1991,Vasigaetal:2001}. Hence we can obtain several
 families of optimal two
dimensional multiply constant-weight codes by
Construction~\ref{abibd}.

\section{Conclusion}

Several new combinatorial constructions for multiply constant-weight
codes are given to yield new infinite families of optimal MCWCs with
small weight or small distance. The Johnson upper bounds of MCWCs
are shown to be asymptotically tight for given weights and distance,
which greatly improves the previous result saying that the bounds
are asymptotically tight to a constant factor and for a smaller
class of MCWCs. Finally, we introduce the concept of two dimensional
MCWCs, for which bounds and constructions are studied.


\bibliographystyle{IEEEtran}
\bibliography{I:/JabRefdata/mine}
%
%

%
%
%
%
%




\end{document}